\documentclass[a4paper,USenglish]{lipics-v2016}

\usepackage{microtype}
\usepackage{amsmath,amssymb,tikz}
\usepackage[T1]{fontenc}
\usepackage{complexity}

\DeclareMathOperator{\Doma}{dom}
\DeclareMathOperator{\Ima}{Im}
\DeclareMathOperator{\Vars}{Vars}
\newcommand{\set}[2]{\ensuremath{\left\{ #1 \mid #2 \right\}}}
\newcommand{\isomorph}{\cong}
\newcommand{\N}{\ensuremath{\mathbb{N}}}
\newcommand{\gr}[2][]{\text{gr}_{#1}(#2)}

\DeclareMathOperator{\ltp}{\ensuremath{<}}
\DeclareMathOperator{\addp}{\ensuremath{+}}
\DeclareMathOperator{\mulp}{\ensuremath{\times}}
\newcommand{\ccg}[1][]{\ensuremath{\ComplexityFont{G}_{#1}}}

\theoremstyle{plain}
\newtheorem{conjecture}[theorem]{Conjecture}
\newtheorem{fact}[theorem]{Fact}

\title{On the Implicit Graph Conjecture}

\author[1]{Maurice Chandoo}
\affil[1]{Leibniz Universität Hannover, Theoretical Computer Science,\\
    Appelstr. 4, 30167 Hannover, Germany\\
    \texttt{chandoo@thi.uni-hannover.de}}
\authorrunning{M. Chandoo} 

\Copyright{Maurice Chandoo}


\begin{document}

\maketitle

\begin{abstract}
The implicit graph conjecture states that every sufficiently small, hereditary graph class has a labeling scheme with a polynomial-time computable label decoder. 
We approach this conjecture by investigating classes of label decoders defined in terms of complexity classes such as P and EXP. For instance, GP denotes the class of graph classes that have a labeling scheme with a polynomial-time computable label decoder. Until now it was not even known whether GP is a strict subset of GR where R is the class of recursive languages. We show that this is indeed the case and reveal a strict hierarchy akin to classical complexity. We also show that classes such as GP can be characterized in terms of graph parameters. This could mean that certain algorithmic problems are feasible on every graph class in GP.
Lastly, we define a more restrictive class of label decoders using first-order logic that already contains many natural graph classes such as forests and interval graphs. We give an alternative characterization of this class in terms of directed acyclic graphs. 
By showing that some small, hereditary graph class cannot be expressed with such label decoders a weaker form of the implicit graph conjecture could be disproven.
\end{abstract}

\section{Introduction}
The class of interval graphs has at most $2^{\mathcal{O}(n \log n)}$ graphs on $n$ vertices. Neither adjacency matrices nor lists are asymptotically space optimal to represent this class since only $\mathcal{O}(n \log n)$ bits should be used to store a graph on $n$ vertices. However, due to the geometrical representation of this class every vertex of an interval graph can be assigned an interval on a discrete line with $2n$ points. Stated differently, every vertex can be labeled with two numbers between $1$ and $2n$ and adjacency of two vertices can be determined by comparing the four numbers. Storing two such numbers for all $n$ vertices requires $n \log 4n^2$ bits and thus is asymptotically optimal. Labeling schemes, also known as implicit representation, generalize this kind of representations by allowing to store a $\mathcal{O}(\log n)$ long binary label at every vertex such that adjacency between two vertices can be determined by running an algorithm on the two labels. 
We investigate what graph classes can or cannot be represented in such a way when restricting the computational complexity of the function that determines adjacency, also called label decoder.

Let us call a graph class that has at most $2^{\mathcal{O}(n \log n)}$ graphs on $n$ vertices small. A simple counting argument shows that only small graph classes can have labeling schemes. The first question that springs to mind is whether all small graph classes have a labeling scheme. This is not the case as Spinrad shows by giving a small, non-hereditary graph class as counter-example in~\cite[Thm.~7]{spinrad}.
Now, the question becomes whether all small, hereditary graph classes have a labeling scheme; this is known as implicit graph conjecture(IGC).  
This question was already posed more than two decades ago in 1992 by Kannan, Naor and Rudich~\cite{kannan} and has been brought up again by Spinrad~\cite{spinrad}. But despite being such an old question not much is known in this regard. One such result is: every tiny, hereditary graph class admits a labeling scheme with labels of constant length~\cite{scheinerman}. Tiny means that there exist $n_0 \in \mathbb{N}$ and $k < \frac{1}{2}$ such that the class has at most $2^{kn \log n}$ labeled graphs on $n$ vertices for all $n \geq n_0$. This follows from the insight that every tiny, hereditary graph class has only a constant number of twin-free graphs, which makes such classes rather uninteresting. On the other hand, small, hereditary graph classes such as planar or circular-arc graphs can have a rich structure. Candidates for the IGC, i.e.~small, hereditary graph classes for which no labeling scheme is known, are line segment graphs, (unit) disk graphs, $k$-dot product graphs and $k$-sphere graphs~\cite{fiduccia,mcdiarmid,kang}. It is interesting to note that the obvious labeling schemes for line segment and disk graphs using their geometrical representation does not work since coordinates and radii can require an exponential number of bits~\cite{mcdiarmid} unlike in the case of interval graphs.

A different aspect of labeling schemes that has been extensively studied are lower and upper bounds on the label length, i.e.~the constant lurking in $\mathcal{O}(\log n)$, which is related to small universal graphs. A recent result shows that graphs of bounded arboricity $k$ admit a labeling scheme with optimal label length $k \log n + \mathcal{O}(1)$~\cite{alstrup}. Besides, labeling schemes can be generalized in various ways. One variant are distance labeling schemes where one wants to infer the distance between two vertices given their labels~\cite{gavoille}. In~\cite{korman} it was proposed to consider multiple labels instead of only two. Another natural extension is to consider labeling schemes for graph classes that are not small by allowing longer labels while still maintaining the condition of being asymptotically space optimal~\cite{spinrad}. However, here we shall investigate the original variant of this concept.

\paragraph*{Our results.}
For a complexity class $\ComplexityFont{A}$ let $\ccg \ComplexityFont{A}$ denote the class of graph classes that have a labeling scheme where the label decoder can be computed in $\ComplexityFont{A}$ (precise definitions follow). In general, we investigate how choosing various complexity classes for $\ComplexityFont{A}$ affects the class of graph classes $\ccg \ComplexityFont{A}$ that can be represented and how such classes of graph classes can be characterized. 
In section two we argue that $\ccg k \EXP \subsetneq \ccg (k+1) \EXP$ for all $k \geq 1$ by giving a diagonalization argument. A related result for distance labeling schemes can be found in section four of \cite{gavoille}. Additionally, we consider the graph class(es) constructed in the proof as candidate for the implicit graph conjecture. In the third section we show that for every reasonable complexity class $\ComplexityFont{A}$ the class of graph classes $\ccg \ComplexityFont{A}$ can be exactly characterized in terms of a graph parameter. By graph parameter we mean a graph property which maps to the natural numbers such as clique number or tree width. Given such a characterizing graph parameter $\lambda_\ComplexityFont{A}$ for $\ccg \ComplexityFont{A}$ the question of whether a graph class lies in $\ccg \ComplexityFont{A}$ then is equivalent to asking whether it is bounded by $\lambda_\ComplexityFont{A}$. Another consequence of such a characterization is that if for example determining the existence of a Hamiltonian cycle  is fixed-parameter tractable under the parameterization $\lambda_\ComplexityFont{A}$ then for every graph class in $\ccg \ComplexityFont{A}$ this problem can be decided in polynomial time. This means the existence of a labeling scheme can have algorithmic implications. In the last section we define a class of label decoders $\FO$ via first-order logic formulas with arithmetic, i.e.~comparing order, addition and multiplication. Our motivation for introducing this class of label decoders is that the Turing machine model seems too strong to obtain lower bounds. We give upper bounds on the expressiveness of $\ccg \FO$ and its quantifier-free variant. Even if quantifiers, addition and multiplication are disallowed the resulting class $\ccg \FO_{\mathrm{qf}}(\ltp)$ already contains many interesting graph classes such as forests, planar graphs and $k$-interval graphs(also known as multiple interval graphs~\cite{downey}). Lastly, we describe an alternative characterization of $\ccg \FO_{\mathrm{qf}}(\ltp)$ in terms of directed acyclic graphs. 

\paragraph*{Terminology.}
Let $[n] = \{1,2,\dots,n\}$. We write $\log n$ instead of $\lceil \log_2 n \rceil$ and $\exp(n) = 2^n$. Let $\exp^i(n) = \exp(\exp^{i-1}(n))$ for $i \geq 1$ and $\exp^0(n) = n$. The domain and image of a function $f$ are abbreviated by $\Doma(f)$ and $\Ima(f)$ respectively.
We consider only graphs without multiple edges and self-loops and regard undirected graphs as special case of directed ones. 
For a sequence of graphs $G,G_1,\dots,G_m$ on the same vertex set $V$ let us say $G$ is the edge-union of $G_1,\dots,G_m$ if $E(G) = \cup_{i \in [m]} E(G_i)$.
For two graphs $G,H$ we write $G \isomorph H$ to indicate that they are isomorphic. 
We speak of $G$ as unlabeled graph to emphasize that we talk about the isomorphism class of $G$ rather than a specific adjacency matrix of $G$. A graph class is a set of unlabeled graphs, i.e.~closed under isomorphism. A graph class is hereditary if it is closed under taking induced subgraphs. Let $\mathcal{G}$ be the class of all graphs and $\mathcal{G}_n$ is the class of all graphs on $n$ vertices.
A language is a set of words over the binary alphabet $\{0,1\}$. We use complexity class as informal term to mean a set of languages defined in terms of computation and assume that it is countable.  The deterministic Turing machine (TM) is our model of computation when talking about time as resource bound. Let $\ComplexityFont{L}$ denote the complexity class logspace, $\PH$ is the polynomial-time hierarchy, $\R$ is the class of recursive languages and $k\EXP$ is the class of languages computable in time $\exp^k(n^{\mathcal{O}(1)})$ for $k \geq 0$, e.g.~$\ComplexityFont{0}\EXP = \P$. Let $\ComplexityFont{ALL} = \mathcal{P}(\{0,1\}^*)$ be the class of all languages.

\begin{definition}[Labeling scheme]
    A label decoder $F$ is a binary relation over words, i.e.~$F \subseteq \{0,1\}^* \times \{0,1\}^*$. A labeling scheme is a tuple $S = (F,c)$ where $F$ is a label decoder and $c \in \mathbb{N}$ is the label length. A graph $G$ on $n$ vertices is in the class of graphs spanned by $S$, denoted by $G \in \gr{S}$, if there exists a labeling $\ell \colon V(G) \rightarrow \{0,1\}^{c \log n}$ such that for all $u, v \in V(G)$:
    $$ (u,v) \in E(G) \Leftrightarrow (\ell(u),\ell(v)) \in F $$
    We say a graph class $\mathcal{C}$ is represented by (or has) a labeling scheme $S$ if $\mathcal{C} \subseteq \gr{S}$.
\end{definition}

\begin{definition}
    A language $L \subseteq \{0,1\}^*$ induces a label decoder $F_L$ where for all $x,y \in \{0,1\}^*$ with $|x|=|y|$ it holds that $ (x,y) \in F_L \Leftrightarrow xy \in L $.
    
    Let $\ComplexityFont{A}$ be a set of languages and $k \in \N$. A graph class $\mathcal{C}$ is in $\ccg[k]\ComplexityFont{A}$ if there exists a language $L \in \ComplexityFont{A}$ such that $\mathcal{C}$ is represented by $(F_L,c)$ for some $c \leq k$. Analogously, $\mathcal{C}$ is in $\ccg \ComplexityFont{A}$ if $\mathcal{C}$ is in $\ccg[k] \ComplexityFont{A}$ for some $k \in \N$. 
\end{definition}

A class of the form $\ccg \cdot$ is trivially closed under taking subsets, i.e.~if $\mathcal{C} \subseteq \mathcal{C'}$ and $\mathcal{C'} \in \ccg \cdot$ then $\mathcal{C} \in \ccg \cdot$. It follows that $\ccg \cdot$ is closed under intersection as well. However, no $\ccg \cdot$ is closed under complement since the complement of a small graph class is not small. 
For many complexity classes such as $\ComplexityFont{L}$ and $\P$ it is also not hard to show that the classes $\ccg \ComplexityFont{L}$ and $\ccg \P$ are closed under union. 

Here is an example of a language $L$ whose label decoder $F_L$ represents interval graphs: $x_1x_2y_1y_2 \in L$ iff $x_1,x_2,y_1,y_2$ are binary strings of equal length and neither $x_2 < y_1$ nor $y_2 < x_1$ holds where $<$ denotes the lexicographical order. Then the labeling scheme $S = (F_L,4)$ represents interval graphs. Since $L$ can be computed in logspace it follows that interval graphs are in $\ccg[4]\ComplexityFont{L}$.

Using our terminology the implicit graph conjecture can be rephrased as:
\begin{conjecture}[IGC,\cite{kannan}]
    Let $\mathrm{H}$ denote the set of all small, hereditary graph classes.
    $$ \ccg \P \cap \mathrm{H} = \ccg \ComplexityFont{ALL} \cap \mathrm{H} =  \mathrm{H} $$
\end{conjecture}

As of now it is far from clear whether even the second equality holds, i.e.~can every small, hereditary graph class be represented by some labeling scheme, leaving computability issues aside?
This is a graph-theoretic question dealing with the existence of polynomial-sized universal graphs that should be addressed before one can expect to prove the implicit graph conjecture.

\section{Hierarchy of Implicit Representations}
In the previous section we have seen that every language $L$ can be interpreted as label decoder $F_L$. Therefore a set of languages $\ComplexityFont{A}$ can be understood as set of label decoders and $\ccg \ComplexityFont{A}$ denotes the set of graph classes that can be represented by a labeling scheme $(F,c)$ with $F \in \ComplexityFont{A}$ and $c \in \N$. Inclusion carries over to this setting meaning $\ComplexityFont{A} \subseteq \ComplexityFont{B}$ implies $\ccg \ComplexityFont{A} \subseteq \ccg \ComplexityFont{B}$. For separations, however, this is not true, i.e.~there exist $\ComplexityFont{A}, \ComplexityFont{B}$ with $\ComplexityFont{A} \subsetneq \ComplexityFont{B}$ and $\ccg \ComplexityFont{A} = \ccg \ComplexityFont{B}$. Spinrad remarks that it is not known whether restricting the label decoder to be computable in polynomial time versus requiring it to be simply computable makes a difference in terms of the graph classes that can be represented~\cite[p.~22]{spinrad}.  
We resolve this question by applying diagonalization, which yields many of the separations known in the classical setting.  
For the sake of clarity we prove the following class of separations which we deem most interesting with respect to the IGC since it yields the smallest class($\ccg \ComplexityFont{2EXP}$) that can be separated from $\ccg \P$ by this argument:
\begin{theorem}
    $\ccg k \EXP \subsetneq \ccg (k+1) \EXP$ for all $k \geq 1$. 
    \label{thm:hierarchy}
\end{theorem}

The basic idea behind the proof of this statement is the following diagonalization argument. Let $\ComplexityFont{A} = \{F_1,F_2, \dots \}$ be a set of label decoders. Then a labeling scheme in $\ccg \ComplexityFont{A}$ can be seen as pair of natural numbers, one for the label decoder and one for the label length. Let $\tau : \N \rightarrow \N^2$ be a surjective function and $S_{\tau(x)} = (F_y,z)$ with $\tau(x) = (y,z)$.
It follows that for every labeling scheme $S$ in $\ccg \ComplexityFont{A}$ there exists an $x \in \N$ such that $S = S_{\tau(x)}$.  
The following graph class cannot be in $\ccg \ComplexityFont{A}$:
$$ G \in \mathcal{C}_\ComplexityFont{A} \Leftrightarrow G \text{ is the smallest graph on $n = |V(G)|$ vertices s.t.~} G \notin \gr{S_{\tau(n)}} $$ 
where smallest is meant w.r.t.~some order such as the lexicographical one. Note that the order must be for unlabeled graphs. However, an order for labeled graphs can be easily adopted to unlabeled ones.   
Assume $\mathcal{C}_\ComplexityFont{A}$ is in $\ccg \ComplexityFont{A}$ via the labeling scheme $S$. There exists an $n \in \N$ such that $S = S_{\tau(n)}$ and it follows that $\mathcal{C}_\ComplexityFont{A}$ contains a graph on $n$ vertices that cannot be in $S$ per definition, contradiction. Then it remains to show that $\mathcal{C}_\ComplexityFont{A}$ is in the class that we wish to separate from $\ccg \ComplexityFont{A}$. 

For the remainder of this section we formalize this idea in three steps. First, we state the requirements for a pairing function $\tau$ and show that such a function exists. We continue by arguing that the diagonalization graph class $\mathcal{C}_{\ComplexityFont{A}}$ is not contained $\ccg \ComplexityFont{A}$. In the last step we construct a label decoder for $\mathcal{C}_{k \EXP}$ and show that it can be computed in $(k+1)\EXP$. 

\begin{definition}
    A surjective function $\tau : \N \rightarrow \N^2$ is an admissible pairing if 
    \begin{enumerate}
        \item $|\tau^{-1}(y,z)|$ is infinite for all $y,z \in \N$,
        \item $\tau_y(x),\tau_z(x) \in \mathcal{O}(\log x)$ with $\tau(x) = (\tau_y(x),\tau_z(x))$,
        \item $\tau(x)$ is undefined if $x$ is not a power of two, and
        \item $\tau$ is computable in polynomial time given its input in unary.
    \end{enumerate}
    \label{def:pairing}
\end{definition}

Note, that a graph on $n$ vertices gets assigned labels of the same length as a graph on $m$ vertices whenever $\log n = \log m$ (rounded up). The third condition prevents this from happening, i.e.~for all $G \neq H \in \mathcal{C}_{\ComplexityFont{A}}$ it holds that their vertices must have labels of different length.
\begin{lemma}
    There exists an admissible pairing function.
    \label{lem:pairing}
\end{lemma}
\begin{proof}
    Consider the function $\tau(x) = (y,z)$ iff $x = 2^{2^y \cdot 3^z \cdot  5^w}$ for some $w \geq 0$. 
\end{proof}

\begin{definition}
    Let $\ComplexityFont{A}$ be a set of languages, $\prec$ an order on unlabeled graphs and $\tau$ an admissible pairing. The diagonalization graph class of $\ComplexityFont{A}$ is defined as:
    $$ \mathcal{C}_\ComplexityFont{A} = \bigcup_{n \in \Doma(\tau)} \left\{ G \in \mathcal{G}_n    \: \middle|            
    G \text{ is the smallest graph w.r.t.~$\prec$ not in } \gr{S_{\tau(n)}}                                
    \right\} $$
\end{definition}

When we consider the diagonalization graph class of a set of languages we assume the lexicographical order for $\prec$ and the function given in the proof of Lemma~\ref{lem:pairing} for $\tau$.
\begin{lemma}
    For every countable set of languages $\ComplexityFont{A}$ it holds that  $\mathcal{C}_\ComplexityFont{A} \notin \ccg \ComplexityFont{A}$.
    \label{lem:ca_not_in_a}
\end{lemma}
\begin{proof}
As argued in the paragraph after Theorem \ref{thm:hierarchy} it holds that for any labeling scheme $S$ in $\ccg \ComplexityFont{A}$ there exists a graph $G$ that is in $\mathcal{C}_{\ComplexityFont{A}}$ but not in $\gr{S}$ and thus this lemma holds. Since the labeling scheme $S$ is in $\ccg \ComplexityFont{A}$ there exists an $n \in \N$ such that $S = S_{\tau(n)}$ where $S_{\tau(n)} = (F_y,z)$, $\tau(n)=(y,z)$ and  $\ComplexityFont{A} = \{F_1,F_2,\dots \}$. Due to the fact that $|\tau^{-1}(y,z)|$ is infinite it follows that there exists an arbitrarily large $n \in \N$ such that $S = S_{\tau(n)}$.
For $\mathcal{C}_{\ComplexityFont{A}} \setminus \gr{S}$ to be non-empty it must hold that $\gr{S_{\tau(n)}}$ does not contain all graphs on $n$ vertices. By choosing $n$ to be sufficiently large this is guaranteed since $\gr{S_{\tau(n)}}$ is a small graph class.         
\end{proof}

To show that $\mathcal{C}_\ComplexityFont{A}$ is in some class $\ccg \ComplexityFont{B}$ we need to define a labeling scheme $S_{\ComplexityFont{A}} = (F_{\ComplexityFont{A}},1)$ that represents $\mathcal{C}_\ComplexityFont{A}$ and consider the complexity of computing its label decoder.
\begin{definition}
    Let $\ComplexityFont{A}$ be a set of languages. 
    For $G \in \mathcal{C}_{\ComplexityFont{A}}$ let $G_0$ denote the smallest labeled graph with $G_0 \isomorph G$. We define the label decoder $F_{\ComplexityFont{A}}$ as follows. For every $m \in \N$ such that there exists $G \in \mathcal{C}_{\ComplexityFont{A}}$ with $|V(G)| = 2^m$ and for all $x, y \in \{0,1\}^m$ let
    $$ (x,y) \in F_{\ComplexityFont{A}} \Leftrightarrow (x,y) \in E(G_0)  $$
\end{definition}

It can be assumed that $G_0$ has $\{0,1\}^m$ as vertex set. Also, note that $\mathcal{C}_{\ComplexityFont{A}}$ has at most one graph on $n$ vertices for any $n$. Therefore the label decoder $F_{\ComplexityFont{A}}$ is well-defined. It is easy to see that $(F_{\ComplexityFont{A}},1)$ represents $\mathcal{C}_\ComplexityFont{A}$, i.e.~$\mathcal{C}_\ComplexityFont{A} \subseteq \gr{F_\ComplexityFont{A},1}$.

Up to this point the exact correspondence between $y \in  \N$ and the label decoder $F_y$ was not important. In fact, we only required the set of label decoders $\ComplexityFont{A}$ to be countable. To show that the label decoder $F_{k\EXP}$ can be computed in $(k+1)\EXP$ it is important that given $y$ the label decoder $F_y$ from $k \EXP$ can be effectively computed. The following lemma grants this.    
\begin{lemma}
    For every $k\geq 0$ there exists a mapping $f \colon \N \rightarrow \ComplexityFont{ALL}$ such that $\Ima(f) = k\EXP$ and on input $x \in \N$ in binary and $w \in \{0,1\}^*$ the question $w \in f(x)$ can be decided in $\exp^{k+1}(n^{\mathcal{O}(1)})$ time with $n = |w|+ \log x$.
    \label{lem:effective}
\end{lemma}
\begin{proof}
    The lemma essentially states that all TMs running in $k \EXP$ can be simulated in $(k+1) \EXP$. Given the G\"{o}delization of such a TM $M$ and a word $w$ as input the question whether $M$ accepts $x$ can be decided by a TM in $(k+1)\EXP$. 
    Fix a reasonable encoding of TMs as natural numbers, i.e.~given $z \in \N$ then $M_z$ is a TM. Let $f(x) = (y,z) \Leftrightarrow x = 2^y 3^z$. It holds that $y \leq \log x$ for every $z \geq 0$. On input $x \in \N$ and $w \in \{0,1\}^*$ the reference input length is $n = |w| + \log x$. Compute $f(x) = (y,z)$ and then simulate $M_z$ on $w$ for $\exp^k(y|w|^y) \leq \exp^k(n^{n+1}) \in \mathcal{O}(\exp^{k+1}(n^2)) $ steps.  
\end{proof}

\begin{lemma}
    $F_{k\EXP} \in (k+1)\EXP$ for every $k \geq 1$.
    \label{lem:fa_compute}
\end{lemma}
\begin{proof}
    On input $xy$ with $x,y \in \{0,1\}^m$ and $m \geq 1$ compute $\tau(2^m) = (y,z)$. If it is undefined then reject. Otherwise there is a labeling scheme $S_{\tau(2^m)}=(F_y,z)$ and we need to compute the smallest graph $G_0$ on $2^m$ vertices such that $G_0 \notin \gr{S_{\tau(2^m)}}$. If $G_0$ exists we assume that its vertex set is $\{0,1\}^m$ and accept iff $(x,y) \in E(G_0)$. If it does not exist then reject. 
    
    The graph $G_0$ can be computed as follows. Iterate over all labeled graphs $H$ with $2^m$ vertices in order and over all bijections $\ell \colon V(H) \rightarrow \{0,1\}^{zm}$. Check if $H \in \gr{S_{\tau(2^m)}}$ by checking for every pair of vertices $u,v \in V(H)$ if $(\ell(u),\ell(v)) \in F_y \Leftrightarrow (u,v) \in E(H)$. If this condition fails then $G_0 = H$. To query the label decoder $F_y$ the previous lemma can be applied, i.e.~$y$ can be interpreted as encoding of a TM in $k\EXP$ that can be simulated.  
    
    Let us consider the time requirement w.r.t.~$m$. To compute $\tau(2^m)$ we write down $2^m$ in unary and compute $\tau$ in polynomial time w.r.t.~$2^m$ which is in the order $2^{\mathcal{O}(m)}$. To compute $G_0$ there are four nested loops. The first one goes over all labeled graphs on $2^m$ vertices which is bounded by $\exp^2(2m)$. The second loop considers all possible labelings $\ell$ of which there can be at most $\exp(zm)^{\exp(m)} =\exp(\exp(m)zm)  \leq \exp^2(zm^2) \in \exp^2(m^{\mathcal{O}(1)})$; recall that $z$ is polynomially bounded by $m$ due to Definition~\ref{def:pairing}. The other two loops go over all vertices of $H$ meaning $2^m$. By applying Lemma~\ref{lem:effective} the time required to compute $(\ell(u),\ell(v)) \in F_y$ is $\exp^{k+1}(n_0^{\mathcal{O}(1)})$ with $n_0 := 2zm + \log y$. Since $n_0 \in m^{\mathcal{O}(1)}$ this operation can be computed in $(k+1)$-exponential time. In summary, the runtime order of this algorithm is $\exp^{k+1}(m^{\mathcal{O}(1)})$.
\end{proof}

Now, Lemma~\ref{lem:ca_not_in_a} states that $\mathcal{C}_{k\EXP} \notin \ccg k \EXP$ and from Lemma~\ref{lem:fa_compute} it follows that $\mathcal{C}_{k\EXP} \in \ccg (k+1) \EXP$ therefore proving Theorem \ref{thm:hierarchy}. Notice, that this argument fails to show that $\ccg \P \subsetneq \ccg \EXP $ because the runtime to compute the label decoder $F_\P$ is at least double exponential due to the first two loops mentioned in the proof of Lemma \ref{lem:fa_compute}. Can this argument  be modified to separate these two classes as well? This seems rather unlikely. 
Nonetheless, we now know that there exist graph classes that have an implicit representation but a polynomial-time computable label decoder does not suffice to capture them.
\begin{fact}
    If there exists a small, hereditary graph class $\mathcal{C}$ with $\mathcal{C}_{\P} \subseteq \mathcal{C}$ then the implicit graph conjecture is false.    
    \label{fact:gpr}
\end{fact}
For two graph classes $\mathcal{C}$ and $\mathcal{D}$ let us call $\mathcal{D}$ the hereditary closure of $\mathcal{C}$ if $G \in \mathcal{D}$ iff $G$ occurs as induced subgraph of some graph in $\mathcal{C}$.
If the hereditary closure of $\mathcal{C}_{\P}$ is not a small graph class then it follows that the premise of Fact \ref{fact:gpr} is unsatisfiable. Recall that $\mathcal{C}_\P$ is not an unambiguous graph class but depends on the chosen order $\prec$ and pairing $\tau$, which makes it difficult to analyze what kind of graphs are contained in such a class. 

\section{Parameter Characterization}
We consider a graph parameter to be a total function $\lambda \colon \mathcal{G} \rightarrow \N$ and call it natural if the cardinality of its image is infinite. Examples of natural graph parameters are the chromatic number or the diameter. A graph class $\mathcal{C}$ is bounded by a graph parameter $\lambda$ if there exists a $c \in \N$ such that for all $G \in \mathcal{C}$ it holds that $\lambda(G) \leq c$.
We show that for every complexity class $\ComplexityFont{A}$ such that $\ccg \ComplexityFont{A}$ is closed under union there exists a graph parameter that characterizes $\ccg \ComplexityFont{A}$. One interesting aspect of such a characterization is that it might reveal algorithmic implications for graph classes that have a labeling scheme of certain complexity.

\begin{definition}
    Let $\mathbb{C}$ be a set of graph classes and $\lambda$ is a graph parameter. We say $\lambda$ characterizes $\mathbb{C}$ if for every graph class $\mathcal{C}$ it holds that $\mathcal{C} \in \mathbb{C}$ iff $\mathcal{C}$ is bounded by $\lambda$. 
\end{definition}

Let us say a set of graph classes $\mathbb{C}$ is complete if for every graph $G$ there exists a $\mathcal{C} \in \mathbb{C}$ such that $G \in \mathcal{C}$. 

\begin{theorem}
    Let $\mathbb{C}$ be a complete set of graph classes closed under union and subsets with $\mathcal{G} \notin \mathbb{C}$. If there exists a countable subset of $\mathbb{C}$ such that its closure under subsets equals $\mathbb{C}$ then there exists a natural graph parameter that characterizes $\mathbb{C}$.
    \label{thm:parmchr}
\end{theorem}
\begin{proof}
    Let $\mathbb{C}$ be a  set of graph classes that satisfies the above premises and $\mathbb{C}' = \{ \mathcal{C}_1, \mathcal{C}_2, \dots \}$ is the needed countable subset of $\mathbb{C}$. Let $\lambda(G)$ be the minimal $i \geq 1$ such that $G \in \mathcal{C}_i$. Since $\mathbb{C}$ is complete it follows that $\mathbb{C}'$ is complete and thus $\lambda$ is total. 
    Let us define $\mathcal{C}^\lambda_{\leq i}$ as $\set{G \in \mathcal{G}}{ \lambda(G) \leq i}$ and similarly $\mathcal{C}^\lambda_{= i}$. It follows that a class $\mathcal{C}$ is bounded by $\lambda$ iff $\mathcal{C} \subseteq \mathcal{C}^\lambda_{\leq i}$ for some $i \in \N$.
    We now argue that $\lambda$ characterizes $\mathbb{C}$. 
    
    If $\mathcal{C} \in \mathbb{C}$ then there exists an $i \in \N$ such that $\mathcal{C} \subseteq \mathcal{C}_i$. It follows that $\mathcal{C} \subseteq \mathcal{C}^\lambda_{\leq i}$. We show the other direction by induction: if $\mathcal{C} \subseteq \mathcal{C}_{\leq i}^{\lambda}$ then $\mathcal{C} \in \mathbb{C}$ for all $i \in \N$. For $i=1$ it holds that $\mathcal{C} \subseteq  \mathcal{C}_{\leq 1}^{\lambda} = \mathcal{C}_{= 1}^{\lambda} = \mathcal{C}_{1}$. Since $\mathbb{C}$ is closed under subsets it follows that $\mathcal{C} \in \mathbb{C}$. For $i+1$ it holds that $\mathcal{C} \subseteq \mathcal{C}_{\leq i+1}^{\lambda}$ and $\mathcal{C}_{\leq i+1}^{\lambda} = \mathcal{C}_{\leq i}^{\lambda} \cup \mathcal{C}_{= i + 1}^{\lambda}$. By induction hypothesis it follows that $\mathcal{C}_{\leq i}^{\lambda} \in \mathbb{C}$. Since $\mathbb{C}$ is closed under union it remains to argue that $\mathcal{C}_{= i + 1}^{\lambda}$ is in $\mathbb{C}$. This follows by the observation  $\mathcal{C}_{= i + 1}^{\lambda} \subseteq \mathcal{C}_{i+1}$ and $\mathcal{C}_{i+1} \in \mathbb{C}$.   
\end{proof}

Let us examine the premises of Theorem \ref{thm:parmchr} with respect to the class of graph classes that we consider. Every class of the form $\ccg \cdot$ is closed under subsets and for a lot of complexity classes $\ComplexityFont{A}$ it also holds that $\ccg \ComplexityFont{A}$ is closed under union. For completeness a lookup table can be constructed for every singleton graph class. The required countable subset is given by the languages of $\ComplexityFont{A}$.  
In fact, every class of the form $\ccg \cdot$ mentioned in this paper satisfies these premises and therefore has a parameter characterization with the only exception being the class $\ccg \ComplexityFont{ALL}$, which provably has no parameter characterization. Assume $\lambda$ is a characterizing parameter for $\ccg \ComplexityFont{ALL}$ and let $A = \{ \mathcal{C}^{\lambda}_{\leq i} \mid i \in \N \}$. It must hold that for every graph class $\mathcal{C} \in \ccg \ComplexityFont{ALL}$ that it is a subset of some graph class in $A$. However, the diagonalization graph class $\mathcal{C}_A$ of $A$ cannot be a subset of any graph class in $A$ but has a labeling scheme and thus is in $\ccg \ComplexityFont{ALL}$, contradiction.

Consider the algorithmic relevance of such characterizations. Let $P \colon \mathcal{G} \rightarrow \{0,1\}$ be a graph property such as having a Hamiltonian cycle and $\lambda$ is a graph parameter that characterizes the class $\ccg \ComplexityFont{A}$. Assume that $P$ can be decided in time $n^{f(k)}$ on input $G$ with $k = \lambda(G)$ for some computable function $f \colon \N \rightarrow \N$. This can also be stated as $P$ parameterized by $\lambda$ being in the complexity class $\XP$. Then it follows that the property $P$ can be decided in polynomial time on every graph class in $\ccg \ComplexityFont{A}$.
The contra-position of this argument can be used to show that a graph class $\mathcal{C}$ is probably not in $\ccg \ComplexityFont{A}$: if it is $\NP$-hard to decide the property $P$ on a graph class $\mathcal{C}$ then this implies that $\mathcal{C}$ cannot be in $\ccg \ComplexityFont{A}$ unless $\P = \NP$.

Of course, the characterizing parameter derived from the proof of Theorem \ref{thm:parmchr} is not suitable for direct analysis but guarantees existence of such a characterization. However, there is room for different parameter characterizations of the same class as the following equivalence notion shows.
For two graph parameters $\lambda_1,\lambda_2$ let us say that $\lambda_2$ bounds $\lambda_1$, in symbols $\lambda_1 \leq \lambda_2$, if every graph class $\mathcal{C}$ that is bounded by $\lambda_1$ is also bounded by $\lambda_2$. 
If $\lambda_1 \leq \lambda_2$ and $\lambda_2 \leq \lambda_1$  we say $\lambda_1$ and $\lambda_2$ are equivalent.
For example, the maximum degree is bounded by clique number but not vice versa. 
\begin{fact}
    Let $\mathbb{C}_1, \mathbb{C}_2$ be two classes of graph classes and $\lambda_1,\lambda_2$ are respective  characterizing graph parameters.
    $\mathbb{C}_1 \subseteq \mathbb{C}_2$ iff $\lambda_1 \leq \lambda_2$.  
\end{fact}

It follows that two graph parameters are equivalent iff they characterize the same class of graph classes. 
For a complexity class $\ComplexityFont{A}$ let $\lambda_\ComplexityFont{A}$ be a characterizing graph parameter thereof. Hence, comparing the containment relation of two classes $\ccg \ComplexityFont{A}$ and $\ccg \ComplexityFont{B}$ is the same as examining whether $\lambda_\ComplexityFont{A}$ bounds $\lambda_\ComplexityFont{B}$ or vice versa.
The interval number $\lambda_{\mathrm{Intv}}(G)$ of a graph $G$ is the smallest number $k \in \N$ such that $G$ is a $k$-interval graph, see~\cite{downey}. 
From this perspective some of our results can be stated as:
$$ 
\lambda_{\mathrm{Intv}} \lneq 
\lambda_{\FO_{\mathrm{qf}}(\ltp)}   \leq 
\lambda_{\ComplexityFont{L}} \leq  
\lambda_{\P}  \leq 
\lambda_{\EXP} \lneq   
\lambda_{\ComplexityFont{2}\EXP} 
\lneq   
\dots
\lneq   
\lambda_{\R}  
$$
where $\lambda \lneq \lambda'$ means strict containment, i.e.~$\lambda \leq \lambda'$ holds and $\lambda' \leq \lambda$ does not hold. The class $\FO_{\mathrm{qf}}(\ltp)$ is introduced in the next section. 
  
\section{First-Order Definable Label Decoders}
For a given small, hereditary graph class there is no obvious way of showing that this class is not contained in $\ccg \P$ or even $\ccg \ComplexityFont{L}$ as the fact that the IGC still stands open has shown. As a consequence, it is reasonable to look at a more restrictive model of computation for label decoders. From a complexity-theoretic view the circuit class $\AC^0$ is probably  among the first candidates. In this case uniformity issues have to be considered, i.e.~the complexity of an algorithm computing the circuits for each input
length. The strongest uniformity condition, which is the most suitable for lower bounds, leads to the class $\FO_{\mathrm{D}}$ from descriptive complexity defined in terms of first-order logic \cite{immerman}.  However, the domain of discourse in this setting would be the positions of the labels, which is arguably not the most natural choice. Instead we propose the domain to be polynomially many natural numbers and a label consists of a constant number of elements of this domain.
In this setting the labeling scheme for interval graphs can be stated as the formula $\varphi(x_1,x_2,y_1,y_2) = \neg ( x_2 < y_1 \vee  y_2 < x_1 )$; compare this with the example given in the first section. It is also possible to describe $k$-interval graphs or any hereditary graph class with linearly many edges such as bounded arboricity graphs with such formulas.

For $n \geq 1$ let $\mathcal{N}_n$ be the structure that has $[n]$ as universe, the order relation $\ltp$ on $[n]$ and addition as well as multiplication defined as functions:
$$\addp(x,y) = \begin{cases}
x+y &, \text{if } x+y \leq n \\
1 &,   \text{if } x+y > n \\
\end{cases}
 \: \: , \: \: 
 \mulp(x,y) = \begin{cases}
 xy &, \text{if } xy \leq n \\
 1 &,   \text{if } xy > n \\
 \end{cases}
$$
For $\sigma \subseteq \{ \ltp, \addp, \mulp \}$ let $\FO_k(\sigma)$ be the set of first-order formulas with boolean connectives $\neg,\vee,\wedge$, quantifiers $\exists,\forall$ and  $k$ free variables using only equality and the relation and function symbols from $\sigma$. For $\sigma = \{\ltp, \addp, \mulp \}$ we simply write $\FO_k$. Let $\Vars(\varphi)$ be the set of free variables in $\varphi$. Given $\varphi \in \FO_k(\sigma)$, $\Vars(\varphi) = (x_1,\dots,x_k)$ and an assignment $a_1,\dots,a_k \in [n]$ we write $\mathcal{N}_n, (a_1,\dots,a_k) \models \varphi$ if the interpretation $\mathcal{N}_n, (a_1,\dots,a_k)$ satisfies $\varphi$ under the usual semantics of first-order logic. 

\begin{definition}
    A (quantifier-free) logical labeling scheme is a tuple $S=(\varphi,c)$ with a (quantifier-free) formula $\varphi \in \FO_{2k}$ and $c,k \in \N$. A $(c,k)$-labeling for a set $V$ is a function $\ell \colon V \rightarrow [n^c]^k$ and induces the graph $G_S^\ell$ with vertex set $V$ and edges $(u,v)$ if $\mathcal{N}_{n^c}, (\ell(u),\ell(v)) \models \varphi$.     
    Then a graph $G$ is in $\gr{S}$ if there exists a $(c,k)$-labeling $\ell$ for $V(G)$ such that $G = G_S^\ell$.
    \label{def:lls}
\end{definition}
\begin{definition}
    Let $\sigma \subseteq \{ \ltp, \addp, \mulp \}$, $c,k \in \N$. A graph class $\mathcal{C}$ is in $\ccg[c,k]\FO(\sigma)$ if there exists a logical labeling scheme $(\varphi,c)$ with $\varphi \in \FO_{2k}(\sigma)$ such that $\mathcal{C} \subseteq \gr{\varphi,c}$. And $\ccg \FO(\sigma) = \cup_{c,k \in \N} \ccg[c,k] \FO(\sigma)$. Let $\ccg\FO_{\mathrm{qf}}(\sigma)$ denote the quantifier-free analogue.
\end{definition}

Notice, $k$ numbers in $[n^c]$ can be encoded as string of length $ck\log n$.
A logical labeling scheme can for instance express a system of polynomial inequalities on $2k$ variables and adjacency is determined by whether this system is satisfied when plugging in the values for two vertices. By disallowing multiplication these systems become linear. Quantified variables can be used to incorporate unknowns. For example, $\varphi(x,y) = \exists z : x \times z^2 = y $ means that there is an edge from $u$ to $v$ with labels $x_u$, $y_v$ if  $y_v$ can be written as product of $x_u$ and a square number. 

\begin{theorem}
    $\ccg \FO \subseteq \ccg \PH$ and $\ccg \FO_{\mathrm{qf}} \subseteq \ccg \ComplexityFont{L}$.
\end{theorem}
\begin{proof}[Proof sketch]
    It is known that the circuit class $\TC^0 \subseteq  \ComplexityFont{L}$ (assuming logspace-uniformity or stronger) and therefore $\ccg \TC^0 \subseteq \ccg\ComplexityFont{L}$~\cite{vollmer}. We argue that $\ccg\FO_{\mathrm{qf}} \subseteq \ccg \TC^0$. Given a logical labeling scheme $(\varphi,c)$ with $\varphi \in \FO_{2k}$ the label length in a graph with $n$ vertices is $ck \log n$. The $\TC^0$-circuit has $2ck \log n$ input bits and every block of $c\log n$ bits corresponds to the value of a free variable in $\varphi$. Every term in $\varphi$ can be evaluated by implementing its syntax tree as part of the circuit since addition and multiplication can be computed in $\TC^0$. The overflow condition, i.e.~if the result is larger than $n^c$,  has to be checked. Then for every atomic formula in $\varphi$ it remains to test for equality or less than of the input terms. 
    After replacing every atomic formula in $\varphi$ by its truth value the formula becomes a propositional formula that can be seen as circuit since it is quantifier-free. If $\varphi$ contains quantifiers assume that it is in prenex normal form, i.e.~$\varphi = Q_1 z_1 \dots Q_{q} z_{q} \psi(x_1,\dots,x_{2k},z_1,\dots,z_k)$ where $Q_i \in \{\exists, \forall \}$ and $\psi$ is a quantifier-free formula. The values for $x_1,\dots,x_{2k}$ are determined by the input string and the value of a variable $z_i$ corresponds to a binary word of length $k \log n$, which is linear in the size of the input string. Using the non-determinism of the polynomial-time hierarchy the values of the $z_i$'s can be ``guessed'' and then evaluated using the $\TC^0$-circuit described before, which can be simulated in polynomial time.        
\end{proof}

Indeed, all of the graph classes mentioned in the beginning of this section are already contained in $\ccg \FO_\mathrm{qf}(\ltp)$. Therefore let us consider this class more closely.
\begin{fact}
    The interval number $\lambda_{\mathrm{Intv}}$ is strictly bounded by a graph parameter that characterizes $\ccg\FO_{\mathrm{qf}}(\ltp)$.
    \label{fact:intv}
\end{fact}
\begin{proof}
    This statement is equivalent to saying that $k$-interval graphs are contained in $\ccg\FO_{\mathrm{qf}}(\ltp)$ and there exists a graph class $\mathcal{C} \in \ccg\FO_{\mathrm{qf}}(\ltp)$ that is no subclass of $k$-interval graphs for all $k \geq 1$.
    The containment of $k$-interval graphs in $\ccg\FO_{\mathrm{qf}}(\ltp)$ for every $k$ follows by translating its geometrical representation into a logical labeling scheme as we have done for interval graphs previously. Consider the family of graphs shown in Figure~\ref{fig:family} where $G_{i+1}$ is obtained by appending a new 4-cycle to $G_i$.
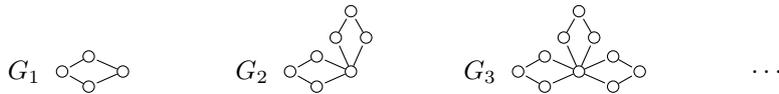
\begin{figure}[b]
    \centering

        \begin{tikzpicture}[shorten >=1pt,auto,node distance=0.1cm,inner sep=0pt,minimum size=0.15cm,
        main node/.style={circle,draw}]

                        \node (x) at (2.5,0) {$\dots$};
        
            \node (a) at (-0.8-0.5,0) {$G_3$};            
        
            \node[main node] (a1) at (0,0) {};
            \node[main node] (a2) at (0.2,0.45) {};
            \node[main node] (a3) at (-0.2,0.45) {};
            \node[main node] (a4) at (0,0.8) {};
            
            \node[main node] (a5) at (0.45,0.2) {};
            \node[main node] (a6) at (0.45,-0.2) {};            
            \node[main node] (a7) at (0.8,0) {};         
            
            \node[main node] (a8) at (-0.45,0.2) {};
            \node[main node] (a9) at (-0.45,-0.2) {};            
            \node[main node] (a10) at (-0.8,0) {};              
        \path[-]
        (a1) edge (a2)
        (a1) edge (a3)
        (a2) edge (a4)
        (a3) edge (a4)
        
        (a1) edge (a5)
        (a1) edge (a6)
        (a5) edge (a7)
        (a6) edge (a7)      
        
        (a1) edge (a8)
        (a1) edge (a9)
        (a8) edge (a10)
        (a9) edge (a10)                     
        ;

            \node[main node] (b1) at (0-3,0) {};
            \node[main node] (b2) at (0.2-3,0.45) {};
            \node[main node] (b3) at (-0.2-3,0.45) {};
            \node[main node] (b4) at (0-3,0.8) {};

            \node[main node] (b8) at (-0.45-3,0.2) {};
            \node[main node] (b9) at (-0.45-3,-0.2) {};            
            \node[main node] (b10) at (-0.8-3,0) {};    
            
            \node (b) at (-0.8-3-0.5,0) {$G_2$};            
            
        \path[-]
        (b1) edge (b2)
        (b1) edge (b3)
        (b2) edge (b4)
        (b3) edge (b4)
        
        (b1) edge (b8)
        (b1) edge (b9)
        (b8) edge (b10)
        (b9) edge (b10)                     
        ;

            \node[main node] (c1) at (0-6,0) {};
            \node[main node] (c8) at (-0.45-6,0.2) {};
            \node[main node] (c9) at (-0.45-6,-0.2) {};            
            \node[main node] (c10) at (-0.8-6,0) {};                
            \node (c) at (-0.8-6-0.5,0) {$G_1$};
            
        \path[-]
        
        (c1) edge (c8)
        (c1) edge (c9)
        (c8) edge (c10)
        (c9) edge (c10)                     
        ;              
        
        \end{tikzpicture}
        \caption{A family of graphs with unbounded interval number}
        \label{fig:family}                          
\end{figure}      
    Then the class $\{ G_i \mid i \in \N \}$ lies in $\ccg\FO_{\mathrm{qf}}(\ltp)$ but can be verified to have unbounded interval number. This follows from the observation that the vertex with maximal degree in $G_i$ cannot be represented with $i-1$ intervals.
\end{proof}

A natural question is how do $c$ and $k$ affect the expressiveness of $\ccg[c,k]\FO_{\mathrm{qf}}(\ltp)$. Non-surprisingly, increasing $k$ strictly enhances the graph classes that can be represented as we will see in a moment. The parameter $c$ determines how large a number stored in a label can be, i.e.~at most $n^c$. In fact, $c$ is degenerate in the sense that it can be bounded in terms of $k$. It would be surprising if the same holds in the presence of addition.
\begin{lemma}
    $\ccg[c,k]\FO_{\mathrm{qf}}(\ltp) \subseteq \ccg[k,k]\FO_{\mathrm{qf}}(\ltp)$ for all $c,k \geq 1$.    
    \label{lem:ckkk}
\end{lemma}
\begin{proof}
    Consider why it suffices for an interval graph on $n$ vertices to use only numbers between $1$ and $2n$ to represent the intervals. For the same reason it makes no difference for a quantifier-free formula $\varphi \in \FO_{2k}(\ltp)$ to be evaluated on a universe larger than $kn$ in the sense that a labeling $\ell \colon V(G) \rightarrow \N^k$ can be converted to a labeling $\ell' \colon V(G) \rightarrow [kn]^k$ such that adjacency is preserved. More precisely, a $(c,k)$-labeling $\ell$ for a vertex set $V$ can be transformed into a $(k,k)$-labeling $\ell'$ such that $G_{(\varphi,c)}^\ell = G_{(\varphi,k)}^{\ell'}$ holds for every quantifier-free formula $\varphi \in \FO_{2k}(\ltp)$. Let $n = |V|$ be the number of vertices. Since $k$ numbers are assigned to each vertex there are at most $kn$ numbers in $A= \{ x_i \mid u \in V, \ell(u) = (x_1,\dots,x_k), i \in [k] \}$. For an $a \in A$ let $\mathrm{ord}(a) = |\{ b \in A \mid b < a \}|+1$, i.e.~the number of numbers in $A$ that are smaller than $a$ plus one. For $u \in V(G)$ we define $\ell'(u)$ as follows. Let $\ell(u) = (x_1,\dots,x_k)$. Then $\ell'(u) = (\mathrm{ord}(x_1),\dots,\mathrm{ord}(x_l))$. Notice that the maximal value for a component of $\ell'(u)$ is $kn$. It remains to check that the truth value of $\varphi$ is invariant under this modified labeling, which follows from the fact that $x < y \Leftrightarrow \mathrm{ord}(x) < \mathrm{ord}(y)$.     
\end{proof}

A consequence of this is that a logical labeling scheme in $\ccg \FO_{\mathrm{qf}}(\ltp)$ is solely determined by its formula $\varphi$. Therefore we consider a quantifier-free formula $\varphi \in \FO_{2k}(\ltp)$ to be the logical labeling scheme $(\varphi,k)$ as well. To check whether a graph $G$ is in $\gr{\varphi}$ it suffices to find a labeling $\ell \colon V(G) \rightarrow \N^k$ with $2k = |\Vars(\varphi)|$ which can be regarded as $(c,k)$-labeling for a sufficiently large $c$. Stated differently, one does not need to worry about the numbers being polynomially bounded.

Also, it implies that for every $k$ there exists a $k' > k$ such that $\ccg[k,k]\FO_{\mathrm{qf}}(\ltp) \subsetneq \ccg[k',k']\FO_{\mathrm{qf}}(\ltp)$. Assume the opposite, then $\ccg  \FO_{\mathrm{qf}}(\ltp)$ collapses to $\ccg[k,k]\FO_{\mathrm{qf}}(\ltp)$. It follows that every graph class in $\ccg  \FO_{\mathrm{qf}}(\ltp)$ can be represented using $k^2 \log n$ bits and therefore has at most $\exp(k^2\log n)$ graphs on $n$ vertices, which obviously cannot be the case for any $k \in \N$. 

\begin{lemma}
    The graph class that is the union of every graph class in $\ccg[k,k]\FO_{\mathrm{qf}}(\ltp)$ is contained in $\ccg\FO_{\mathrm{qf}}(\ltp)$ for all $k \in \N$.
    \label{lem:degn}
\end{lemma}
\begin{proof}
    We argue that $\ccg\FO_{\mathrm{qf}}(\ltp)$ is closed under finite union and that there exists only a finite number of labeling schemes in $\ccg[k,k]\FO_{\mathrm{qf}}(\ltp)$ such that they represent different graph classes. For closure under union consider two labeling schemes given by their quantifier-free formulas $\varphi,\psi \in \FO_{2k}(\ltp)$. Then the graph class given by the following formula with $2k+2$ variables contains the union of $\gr{\varphi}$ and $\gr{\psi}$:
    \begin{align*}
       &\big( x_{k+1} = x_{2k+2} \Rightarrow \varphi(x_1,\dots,x_k,x_{k+2},\dots,x_{2k+1}) \big) \wedge \\ &\big( x_{k+1} \neq x_{2k+2} \Rightarrow \psi(x_1,\dots,x_k,x_{k+2},\dots,x_{2k+1}) \big) 
    \end{align*}
    The second claim follows from the fact that there are only finitely many semantically different quantifier-free formulas in $\FO_{k}(\ltp)$ for every $k$. More precisely, there are at most $2k^2$ different atomic formulas (`$\ltp$' and `$=$') on $k$ variables and therefore at most $\exp^2(2k^2)$ semantically different formulas, which is the number of boolean functions on $2k^2$ variables.  
\end{proof}

\begin{definition}
    For a graph $G$ and $k \in \N$ we define the graph parameter $\lambda_{\FO_{\mathrm{qf}}(\ltp)}$ such that $\lambda_{\FO_{\mathrm{qf}}(\ltp)}(G) = k$ if $k$ is the minimal number with $\{G\} \in \ccg[k,k]\FO_{\mathrm{qf}}(\ltp)$.
\end{definition}

\begin{fact}
    The graph parameter $\lambda_{\FO_{\mathrm{qf}}(\ltp)}(G)$ characterizes $\ccg \FO_{\mathrm{qf}}(\ltp)$.
\end{fact}
\begin{proof}
    One direction is trivial: if $\mathcal{C}$ is in $\ccg \FO_{\mathrm{qf}}(\ltp)$ then it is bounded by $\lambda_{\FO_{\mathrm{qf}}(\ltp)}$. For the other direction let $\mathcal{C}$ be bounded by $\lambda_{\FO_{\mathrm{qf}}(\ltp)}$ meaning that there exists a $k$ such that for every $G \in \mathcal{C}$ it holds that $\lambda_{\FO_{\mathrm{qf}}(\ltp)}(G) \leq k$. Therefore $\mathcal{C}$ is a subset of the union of all graph classes in $\ccg[k,k]\FO_{\mathrm{qf}}(\ltp)$ which is in $\ccg \FO_{\mathrm{qf}}(\ltp)$ by Lemma \ref{lem:degn}.
\end{proof}

We remark that a similar construction using the label length does not yield a characterizing parameter for $\ccg \P$ or $\ccg \ComplexityFont{L}$. More specifically, the parameter defined by $\lambda(G) = $ minimal $k$ such that $\{G\} \in \ccg[k]\P$ does not characterize $\ccg \P$ simply because the union of all graph classes in $\ccg[1]\P$ already contains all graphs(the analogon of Lemma \ref{lem:degn} fails).   

\subsection*{Directed Acyclic Graph Characterization}
The semantics of a logical labeling scheme given by a quantifier-free formula $\varphi \in \FO_{2k}(\ltp)$ can be alternatively characterized by directed acyclic graphs (DAGs). Intuitively, an edge  in the DAG corresponds to an atomic formula using `$<$'. The atomic formulas involving equality can be modeled by grouping variables together. This means the DAG has not the variables of $\varphi$ as vertex set but rather a partition of these variables.  

\begin{definition}
    Let $k \in \N$. We call a DAG $D=(X,E_D)$ a $k$-DAG if its vertex set $X$ partitions $[2k]$. A $k$-labeling of a vertex set $V$ is a function $\ell \colon V \rightarrow \N^k$. A $k$-DAG $D$ and a $k$-labeling $\ell$ of a vertex set $V$ define the graph $G_D^\ell$ on vertex set $V$ with the following edges. For $u,v \in V$ let $(\ell(u),\ell(v)) = (x_1,\dots,x_{2k})$. There is an edge $(u,v)$ in $G_D^\ell$ if the following two conditions are satisfied:
    \begin{enumerate}
        \item For all $i,j \in [2k]$ it holds that $x_i = x_j$ whenever $i,j$ are in the same part of $X$,
        \item For all edges $(A,B) \in E_D$ it holds that $x_i < x_j$ for all $i \in A$ and $j \in B$.
    \end{enumerate} 
    \label{def:kdag}
\end{definition}

\begin{definition}
    A graph $G=(V,E)$ is $k$-expressible for a $k \in \N$ if there exists a finite sequence of $k$-DAGs $D_1,\dots,D_r$ and a $k$-labeling $\ell$ of $V$ such that 
    $G$ is the edge-union of $G_{D_1}^\ell,\dots,G_{D_r}^\ell$.
\end{definition}

\begin{theorem}
    For a graph $G$ and $k \in \N$ it holds that $\lambda_{\FO_{\mathrm{qf}}(\ltp)}(G) = k$ iff $k$ is the minimal number such that $G$ is $k$-expressible.
\end{theorem}
\begin{proof}
    We show that there is a one-to-one correspondence between the semantics of a quantifier-free formula $\varphi \in \FO_{2k}(\ltp)$ and $k$-DAGs. We can assume that $\varphi$ contains no negation. To see that this can be done without loss of generality let $\varphi$ be in negation normal form. Then $\neg x = y$ can be replaced by $x < y \vee y < x$ and $\neg x < y$ by $y < x \vee x = y$. Next, we assume that $\varphi$ is in disjunctive normal form, i.e.~$\varphi = C_1 \vee \dots \vee C_p$ where $C_i$ consists of atomic formulas linked by conjunction. Given a $(c,k)$-labeling $\ell$ for a vertex set $V$ the formula $\varphi$ induces the graph $G_S^\ell$ with $S=(\varphi,k)$ as described in Definition \ref{def:lls}. Due to the observation given after the proof of Lemma \ref{lem:ckkk} it is okay to consider a less restrictive $k$-labeling $\ell \colon V \rightarrow \N^k$ instead and additionally we write $G^\ell_\varphi$ instead of $G_S^\ell$. Since every clause $C_i$ is a formula as well it can be seen as logical labeling scheme, which induces the graph $G_{C_i}^\ell$. Then the correspondence between the graphs induced by $\varphi$ and its clauses $C_1,\dots,C_p$ is that $G_\varphi^\ell$ is the edge-union of $G_{C_1}^\ell,\dots, G_{C_p}^\ell$. If a clause is unsatisfiable then its induced graph is the empty graph and thus removing this clause does not affect $G_S^\ell$. Therefore we assume that every clause is satisfiable.
    
    We now argue how to convert a clause $C$ from $\varphi$ into a $k$-DAG $D=(X,E_D)$ such that $G_{C}^\ell = G_{D}^\ell$ for every $k$-labeling $\ell$. Consider the undirected graph $H$ which has the variables of $\varphi$ as vertices and two vertices $x_i,x_j$ are adjacent     
    iff the clause $C$ contains $x_i = x_j$ or $x_j = x_i$. It follows that the connected components of $H$ partition the variables of $\varphi$; let $X$ be this partition. Now, consider the directed graph $F$ which has the variables of $\varphi$ as vertices again and there is an edge $(x_i,x_j)$ in $F$ iff $C$ contains the atomic formula $x_i < x_j$. Since we can assume $C$ to be satisfiable it follows that for every part $A$ in the partition $X$ ($A$ is a subset of the variables of $\varphi$) the induced subgraph of $F$ on the vertex set $A$ yields the independent graph. Assume the opposite, then there exist two variables $x_i,x_j$ in the same part of $X$ such that $(x_i,x_j)$ is an edge in $F$. This means that $C$ contains the atomic formulas $x_i = x_j$ and $x_i < x_j$, which contradicts satisfiability of $C$. Let us define the operation of merging a set of vertices $S$ in a graph $G$ such that the resulting graph $G'$ is the same as $G$ except that all vertices in $S$ are replaced by a single vertex $v_s$ and there is an edge $(u,v_s)$ in $G'$ if there is a vertex $v \in S$ such that $(u,v)$ is an edge in the old graph $G$; analogously for edges $(v_s,u)$. Now, let $F'$ be the graph obtained from $F$ by merging each part of $X$. Then there is a natural one-to-one correspondence between the partition $X$ and the vertex set of $F'$. We define $D$ to have the same edges as $F'$ via this correspondence. It remains to check that for this construction $G_{C}^\ell = G_{D}^\ell$ holds indeed.     
    To prove the other direction a $k$-DAG can be converted into a conjunctive clause in a similar way.     
\end{proof}
We conclude with the following two observations. By adding edge weights $w \colon E \rightarrow \N$ to the $k$-DAGs and adjusting the second condition of Definition \ref{def:kdag} such that for all edges $(A,B) \in E_D$ it holds that $x_j - x_i \geq w(x_i,x_j)$ for all $x_i \in A, x_j \in B$ the semantics of existential quantifiers can be mimicked. Besides, given two $k$-DAGs $D_1$ and $D_2$ with identical vertex sets $V(D_1)=V(D_2)$ it holds that $G_{D_1}^\ell = G_{D_2}^\ell$ for every $k$-labeling $\ell$ whenever the transitive closures of $D_1$ and $D_2$ coincide.
\section{Conclusions and Future Research} 
We have seen that limiting the computational resources for label decoders does indeed affect the class of graph classes that can be represented. Unfortunately, for a specific graph class the diagonalization argument from the second section does not help us determine whether it lies in $\ccg \P$. However, as of now it is not even clear whether any candidate of the IGC admits a labeling scheme at all as mentioned at the end of the first section. Therefore trying to place any of these classes in $\ccg \P$ seems elusive. On the other side, proving lower bounds against $\ccg \P$ or $\ccg \ComplexityFont{L}$ for small, hereditary graph classes might be just as futile given the lack of a suitable reduction notion. To counter this grim situation we have introduced a logical framework
in the previous section that is much more restrictive than the TM model in its quantifier-free variant but still expressive enough to capture many of the implicit representations that we know. It appears to be a realistic goal to prove impossibility results in this setting, or more concretely refute the following weaker version of the IGC: 
\begin{conjecture}[Weak IGC]
    Every small, hereditary graph class is in $\ccg \FO_{\mathrm{qf}}$.
\end{conjecture}
As a first step in this direction we have investigated the fragment $\ccg \FO_{\mathrm{qf}}(\ltp)$ and made some structural observations. With the concept of parameter characterizations we have shown that the question of whether a certain graph class lies in $\ccg \FO_{\mathrm{qf}}(\ltp)$ can be answered by considering the $k$-expressibility property of every graph in this class independently. The directed acyclic graph characterization gives an alternative view on $\ccg \FO_{\mathrm{qf}}(\ltp)$,  which is independent of the logical formalism. This could be a useful tool for proving lower bounds against this class. But even this small fragment seems to be surprisingly expressive as the following task shows. Give an example of a family of graphs that is not bounded by $\lambda_{\FO_{\mathrm{qf}}(\ltp)}$. Recall that for the interval number this was quite simple, see Figure~\ref{fig:family}. Another interesting question is whether adding quantifiers enhances the expressiveness, i.e.~$\ccg \FO_{\mathrm{qf}}(\ltp) = \ccg \FO(\ltp)$?

\subparagraph*{Acknowledgments}
We thank the anonymous reviewers for their helpful comments on earlier drafts of this paper.


\end{document}